\documentclass[conference, twocolumn]{IEEEtran}
\IEEEoverridecommandlockouts
\usepackage{caption}
\usepackage{subcaption}
\usepackage{cite}
\usepackage{url}
\usepackage{amsmath,amssymb,amsfonts}
\usepackage{amsthm}
{
	\newtheorem{prob}{Problem}

	\newtheorem{definition}{Definition}
	
	\newtheorem{theorem}{Theorem}
	
}
\usepackage{algorithmic}
\usepackage{graphicx}
\usepackage{textcomp}
\usepackage{xcolor}
\def\BibTeX{{\rm B\kern-.05em{\sc i\kern-.025em b}\kern-.08em
    T\kern-.1667em\lower.7ex\hbox{E}\kern-.125emX}}
\begin{document}

%
\title{Minimizing Age of Information via Hybrid NOMA/OMA
}

\author{\IEEEauthorblockN{Qian Wang\textsuperscript{1,2}, He Chen\textsuperscript{2}, Yonghui Li\textsuperscript{1}, Branka Vucetic\textsuperscript{1}}%
	\thanks{This work is done when Qian Wang is a visiting student at the Chinese University of Hong Kong.}
	\IEEEauthorblockA{\textsuperscript{1}School of Electrical and Information Engineering, The University of Sydney, Sydney, Australia} 
	\IEEEauthorblockA{\textsuperscript{2} Department of Information Engineering, Chinese University of Hong Kong, Hong Kong SAR, China\\
		\textsuperscript{1}\{qian.wang2,  yonghui.li, branka.vucetic\}@sydney.edu.au, \textsuperscript{2}he.chen@ie.cuhk.edu.hk}}

\maketitle

\begin{abstract}
This paper considers a wireless network with a base station (BS) conducting timely transmission to two clients in a slotted manner via hybrid non-orthogonal multiple access (NOMA)/orthogonal multiple access (OMA). Specifically, the BS is able to adaptively switch between NOMA and OMA for the downlink transmission to minimize the information freshness, characterized by Age of Information (AoI), of the network. If the BS chooses OMA, it can only serve one client within a time slot and should decide which client to serve; if the BS chooses NOMA, it can serve both clients simultaneously and should decide the power allocated to each client. To minimize the weighted sum of expected AoI of the network, we formulate a Markov Decision Process (MDP) problem and develop an optimal policy for the BS to decide whether to use NOMA or OMA for each downlink transmission based on the instantaneous AoI of both clients. We prove the existence of optimal stationary and deterministic policy, and perform action elimination to reduce the action space for lower computation complexity. The optimal policy is shown to have a switching-type property with obvious decision switching boundaries. A suboptimal policy with lower computation complexity is also devised, which can achieve near-optimal performance according to our simulation results. The performance of different policies under different system settings is compared and analyzed in numerical results to provide useful insights for practical system designs.
\end{abstract}
%
%
\section{Introduction}
Recently, researchers have shown an enormous interest \cite{kaul2012real,wang2019minimizing2,ceran2019average,wang2018skip,wang2019minimizing,kaul2012status,gu2019timely,gu2019minimizing,wang2019age,kadota2018optimizing,kadota2018scheduling,kadota2019minimizing,yates2017status,jiang2018can,maatouk2019minimizing1} in a new performance metric, termed \textit{Age of Information} (AoI), due to its advantages in characterizing the timeliness of data transmission in status update systems. The timeliness of status update is of great importance, especially in real-time monitoring applications, in which the dynamics of the monitored processes need to be well grasped for further actions. The AoI is defined as the time elapsed since the generation time of the latest received status update at the destination \cite{kaul2012real}. According to the definition, the AoI is jointly determined by the transmission interval and the transmission delay. Early work on the analysis and optimization of AoI in various networks has mainly focused on the simple single-source system model \cite{ceran2019average,wang2019minimizing2,wang2018skip,wang2019minimizing,kaul2012real,kaul2012status,gu2019timely,gu2019minimizing,wang2019age}. 

Recent researches on AoI optimization pays more attention in the more practical multi-source systems \cite{yates2017status,kadota2018optimizing,kadota2018scheduling,kadota2019minimizing,jiang2018can,maatouk2019minimizing1}. For systems with multiple sources, the AoI of the whole system heavily depends on the transmission scheduling of all devices. In this line of research, the authors in \cite{kadota2018optimizing} considered a base station (BS) receiving status updates from multiple nodes with a \textit{generate-at-will} status arrival model. A BS serving status updates to multiple nodes with randomly generated status update was considered in \cite{kadota2019minimizing}. Both of them derived the lower bound of the weighted sum of the expected AoI of the considered network and compared the lower bound with that of various suboptimal policies including Whittle index policy, max-weight policy, etc. The authors in \cite{jiang2018can} also considered system with stochastic arrivals and derived the Whittle index policy in closed form. A decentralized policy was further proposed in \cite{jiang2018can}, which was shown to achieve near-optimal performance. Another branch of this research line is to optimize the AoI of the networks with random access protocols. Particularly, the AoI performance of slotted ALOHA and Carrier Sense Multiple Access (CSMA) were investigated in \cite{yates2017status,maatouk2019minimizing1}, respectively. 

All aforementioned studies on AoI have concentrated on the orthogonal multiple access (OMA) scheme. That is, only a status update packet can be delivered and received in one time slot. Very recently, the authors in \cite{maatouk2019minimizing} have for the first time investigated the potential of applying non-orthogonal multiple access (NOMA) in minimizing the average AoI of a two-node network. The results in \cite{maatouk2019minimizing} showed that OMA and NOMA can outperform each other in different setups. NOMA has been regarded as a promising technique to deal with large-scale Internet of thing (IoT) deployment \cite{ding2017application,saito2013non}. The idea of NOMA is to leverage the power domain to enable multiple clients to be served at the same time or frequency band. Compared to OMA, NOMA can reduce AoI by improving spectrum utilization efficiency. Specifically, more than one client can be served by the BS using NOMA, resulting in a possible AoI drop of more than one client. While in OMA, only the served client may have AoI drop and the AoI of other clients all increases. In this context, a natural question arises: how should a multi-user system adaptively switch between OMA and NOMA modes to minimize the AoI of the network? To the best of authors' knowledge, the answer to this question remains unknown in the literature. Note that the NOMA scheme makes it possible for the BS to serve more clients at the cost of a high error probability, which makes the optimal policy design problem for hybrid NOMA/OMA non-trivial.


Motivated by the gap above, in this paper, we consider a wireless network with a BS that conducts timely transmission to two clients in a slotted manner. The BS is able to adaptively switch between NOMA and OMA for the downlink transmission. To achieve the optimal AoI performance, the BS needs to decide which scheme (i.e., NOMA or OMA) to use at the beginning of each time slot. For the OMA scheme, the BS should further decide which client to serve. For the NOMA scheme, the BS needs to further decide the power allocated to each client. We make the following contributions in this paper: To minimize the AoI of the network, we develop an optimal policy for the BS to decide whether to use NOMA or OMA for each downlink transmission based on the instantaneous AoI of both clients by formulating a Markov Decision Process (MDP) problem. We prove the existence of the optimal stationary and deterministic policy, and perform action elimination to reduce the action space for lower computation complexity. The optimal policy is shown to have a switching-type property with obvious decision switching boundaries. A suboptimal policy with lower computation complexity is also proposed, which can achieve near-optimal performance according to our simulation results. The approximate average AoI performance of the suboptimal policy is also derived by applying a two-dimensional Markov chain. The performance of different policies under different system settings is compared and analyzed in numerical results to provide useful insights for practical system designs.

\section{System Model and Problem Formulation}

We consider a wireless network with a BS that conducts timely transmission to two clients in a slotted manner. At the beginning of each time slot, the BS can generate a packet for each client, which is known as \textit{generate-at-will} in the literature. This model is practical as the BS generally has the ability to control when to download information from cloud or server. Adaptive NOMA/OMA transmission scheme is adopted by the BS. That is, the BS adaptively switches between NOMA and OMA for the downlink transmission. Thus, it is possible for two clients to receive their packets simultaneously within one time slot. At the end of each time slot, if client $i$ has received its packet successfully from the BS, it will send an acknowledgment (ACK) to the BS. The ACK link from both clients to the base station is considered to be error-free and delay-free. 

We adopt a recent metric, \textit{Age of Information} (AoI) \cite{kaul2012real}, to characterize the timeliness of the information received at each client. It is defined as the time elapsed since the generation time of the latest received information at the destination side. Mathematically, the AoI of client $i$, denoted by $\Delta_i(t)$, at time $t$ is $t-u_i(t)$, where $u_i(t)$ denotes the generation time of latest received status update at time $t$. According to the \textit{generate-at-will} status generation model at the BS, if client $i$ has received its information from the BS, its AoI will decrease to $1$, otherwise its AoI increases by $1$. Mathematically, we have
\begin{equation}
\Delta_i(t+1)=
\left\{
\begin{array}{rcl}
\Delta_i(t)+1 ,& &v_i(t)=0 ,\\
1,& &v_i(t)=1, \\
\end{array}
\right.
\end{equation} where $v_i(t)$ denotes the indicator that is equal to $1$ when the client $i$ receives its information from the BS in time slot $t$ and $v_i(t)=0$ otherwise. The weighted sum of the expected AoI of the two clients is adopted to measure the network-wide information timeliness, which is given by
\begin{equation}
\bar{\boldmath \Delta}=\lim_{T \rightarrow \infty}\frac{1}{T}\mathbb{E}[\sum_{i=1}^{2}\sum_{t=1}^{T}w_i\Delta_i(t)],
\end{equation} where $w_i$ is the weight coefficient of client $i$ with $\sum_{i=1}^2 w_i=1$, and the expectation is taken over all possible system dynamics.

We consider that in the OMA mode, the BS allocates time slots to conduct transmission to each client individually. In this context, if the time slot is assigned for transmission to client $i$, $i\in\{1,2\}$, the observation at the client $i$ can be written as 
\begin{equation}
y_i=h_i \sqrt{P}m_i+n_i,
\end{equation} where $P$ is the constant transmission power of the BS; $m_i$ is the status update information from the BS to client $i$; $h_i$ is the channel coefficient between the BS and client $i$. Specifically, $h_i=\sqrt{d_i^{-\tau}}g_i$, where the normalized distance $d_i=c_i/c_0$, with $c_i$ and $c_0$ denoting the distance between client $i$ to the BS and the baseline distance, respectively. $\tau$ denotes the path loss exponent and $g_i \sim \mathcal{CN}(0,1)$ ($\mathcal{CN}$ denotes complex normal distribution). Without loss of generality, we consider $c_1<c_2$, i.e., ${|h_1|}^2>{|h_2|}^2$. $n_i$ is the complex additive Gaussian noise with variance $\sigma_i^2$. For simplify, we assume the variance of $n_i$ is identical for both clients, i.e., $\sigma_i^2=\sigma^2$, $\forall i$. After receiving the signal, the information can be decoded in an interference-free manner with a SNR $\gamma_i={|h_i|}^2\rho$, where $\rho=P/\sigma^2$ is the transmission SNR. Then, the rate for client $i$ can be expressed as $R_i^{OMA}=\log(1+\gamma_i)$. The outage probability at client $i$ using OMA is given by
\begin{small}
\begin{equation}
\begin{split}
P_i^{O}&=1-P\left(R_i^{OMA}\geq R_i\right)\\
&=1-\exp\left(-\frac{(2^{R_i}-1)d_i^{\tau}}{\rho}\right),
\end{split}
\end{equation} 
\end{small} where $R_i$ is the target rate of client $i$. For simplicity, we assume that the target rates of both clients are the same, i.e., $R_1 = R_2 = R$. Note that the framework developed in this paper can be readily extended to the case with distinct target rates. 

On the other hand, in NOMA, the signals to different clients is combined in the power domain by allocating different power levels to them at the BS. Thus, through successive interference cancellation (SIC), it is possible for two clients to successfully recover their corresponding information in the same time slot. As we consider fixed power transmission, the observation at client $i$ can be given by 
\begin{equation}
y_i=h_i (\sqrt{\alpha_1 P}m_1+\sqrt{\alpha_2 P}m_2)+n_i,
\end{equation} where $\alpha_i$ is the power allocation coefficient, and we readily have $\alpha_1+\alpha_2=1$ to achieve the best possible performance. It is assumed that the BS only has the knowledge of stochastic channel state information (CSI) of its channels to both clients, while the clients as receivers have perfect knowledge of CSI as in \cite{cui2016novel,yu2017performance}. In this way, we have $\alpha_1<\alpha_2$ according to the NOMA principle. 

Then, for client $2$ (i.e., the far user), it decodes its message from the BS directly by treating $m_1$ as interference and its received SINR is can be written as
\begin{equation}
\gamma_{22}=\alpha_2 {|h_2|}^2/(\alpha_1 {|h_2|}^2+1/\rho).
\end{equation} Therefore, the outage probability of client $2$ using NOMA is given by
\begin{equation}
\label{outage2}
\begin{split}
P_2^{N}&=1-P(\log(1+\gamma_{22})\geq R)\\
&=1-\exp\left(-\frac{(2^{R}-1)d_2^\tau}{\rho (\alpha_2-\alpha_1(2^{R}-1))}\right),
\end{split}
\end{equation} where we enforce  $\alpha_2-\alpha_1(2^{R}-1)>0$, i.e., $\alpha_2>\frac{2^R-1}{2^R}$.

For client $1$ (i.e., the near user), it will conduct SIC. Specifically, client $1$ will first decode $m_2$ as what client $2$ has done by treating $m_1$ as interference. The received SINR of client $1$ when decoding $m_2$, $\gamma_{12}$ can thus be similarly expressed as
\begin{equation}
\gamma_{12}=\alpha_2 {|h_1|}^2/(\alpha_1 {|h_1|}^2+1/\rho).
\end{equation} After $m_2$ is successfully decoded, client $1$ decodes $m_1$ without interference, and the SNR is
\begin{equation}
\gamma_{11}=\alpha_1 {|h_1|}^2\rho.
\end{equation} The outage probability of client $1$ using NOMA can thus be calculated as
\begin{small}
\begin{equation}
\label{outage1}
\begin{split}
&P_1^{N}=1-P(\log(1+\gamma_{22})\geq R\& \log(1+\gamma_{11})\geq R)\\
&=1-\exp\left(-\max\left\{\frac{(2^{R}-1)d_1^\tau}{\rho (\alpha_2-\alpha_1(2^{R}-1))},\frac{(2^{R}-1)d_1^{\tau}}{\rho\alpha_1}\right\}\right).
\end{split}
\end{equation}
\end{small}

Comparing the outage probability of each client using either NOMA or OMA scheme, we can find that NOMA offers more chance for the BS to transmit time-sensitive information to both clients at the cost of a higher outage probability. Thus, to maintain the timeliness of the information received at each client, at the beginning of each time slot, the BS needs to carefully decide whether to use NOMA or OMA scheme.  In addition, the outage probability of NOMA is determined by the power allocation for each client. As such, when using NOMA, the BS should appropriately allocate power for the transmission to each client. The power allocated to each client is considered to be discrete in this paper, which is practical in real systems. Specifically, the power allocated to client $i$, denoted by $p_i$, can only take the value from the discrete set $\{0,p,2p,3p,..Np\}$ with $p=P/N$ and $p_1+p_2=P$, as $\alpha_1=1-\alpha_2$. That is, $\alpha_i$ can take the value from $\{0,\frac{1}{N},\frac{2}{N},\frac{3}{N},..,1\}$. As client $2$ is far from the BS (i.e., $c_1<c_2$), to effectively use NOMA, $\alpha_2$ should be larger than $\alpha_1$ when applying NOMA, i.e., $\alpha_2>0.5$. Combining it with the previous condition $\alpha_2>\frac{2^R-1}{2^R}$, we can deduce that $\alpha_2$ can only take value from $\{0,\max\{\frac{1}{2}+\frac{1}{N},\lceil\frac{(2^R-1)N}{2^R}\rceil\frac{1}{N}\},\max\{\frac{1}{2}+\frac{1}{N},\lceil\frac{(2^R-1)N}{2^R}\rceil\frac{1}{N}\}+\frac{1}{N},...,1\}$.

Let $\alpha_2(t)$ denote the power allocation coefficient for client $2$ at time slot $t$. Specifically, $\alpha_2(t)=0$ and $\alpha_2(t)=1$ indicates the BS uses OMA scheme, conducting transmission to client $1$ and client $2$, respectively; otherwise, the BS uses NOMA scheme, serving both clients with the amount of power $\alpha_2(t)P $ allocated to client $2$ and $(1-\alpha_2(t))P$ to client $1$.


Let $\pi$ denote the transmission policy at the BS, consisting of a sequence of actions at each time slot, denoting by $\{a_t\}$. $a_t\in \{0,\max\{\lceil{\frac{N}{2}\rceil}+1,\lceil\frac{(2^R-1)N}{2^R}\rceil\},...,N\}$ indicate the BS allocate $a_tp$ amount of power to client $2$. If $a_t=0$, the BS chooses OMA scheme and transmits information to client $1$; if $a_t=N$, the BS chooses OMA scheme and transmits information to client $2$;  otherwise, the BS chooses NOMA scheme, with $a_tp$ amount of power allocated to client $2$ and $P-a_tp$ allocated to client $1$. Our design objective is to find the optimal policy for the BS that adaptively switches between NOMA and OMA schemes to minimize the weighted-sum of the expected AoI for both clients. The problem can be formally formulated as
\begin{prob}
	\label{p1}
		\begin{equation}
		\min_{\pi} \bar{\boldmath \Delta}.\\
		\end{equation}
\end{prob} 
\section{Age-Optimal policy}
\subsection{MDP Formulation} 
In this subsection, we solve Problem \ref{p1} by recasting it into a MDP problem, described by a tuple $\{\mathcal{S},\mathcal{A},\mathrm{P},r\}$, where

 \begin{itemize}
	\item State space $\mathcal{S}=\mathcal{N}^2$: The state at time slot $t$ is composed by the instantaneous AoI of both clients, $s_t\triangleq (\Delta_{1t},\Delta_{2t})$.
	\item Action space $\mathcal{A}=\{0,\max\{\lceil{\frac{N}{2}\rceil}+1,\lceil\frac{(2^R-1)N}{2^R}\rceil\},...,N\}$: the detailed description of action $a_t \in \mathcal{A}$ has been provided in the previous section. 
	\item Transition probabilities $\mathrm{P}$: $P(s_{t+1}|s_t,a_t)$ is the probability of transit from state $s_t$ to $s_{t+1}$ when taking action $a_t$. According to the outage probability of both clients using either NOMA or OMA given in Section II, we have
	\begin{equation}
	\label{e3}
	\begin{aligned}
	&P((1,\Delta_2+1)|(\Delta_1,\Delta_2),a=0)=1-P_1^O,\\
	&P((\Delta_1+1,\Delta_2+1)|(\Delta_1,\Delta_2),a=0)=P_1^O,\\
	&P((\Delta_1+1,1)|(\Delta_1,\Delta_2),a=N)=1-P_2^O,\\
	&P((\Delta_1+1,\Delta_2+1)|(\Delta_1,\Delta_2),a=N)=P_2^O,\\	
	\end{aligned}
	\end{equation} and for $i\neq 0,N$
	\begin{equation}
	\begin{aligned}
	\label{e3*}
	&P((1,\Delta_2+1)|(\Delta_1,\Delta_2),a=i)=(1-P_1^N(a))P_2^N(a),\\
	&P((\Delta_1+1,1)|(\Delta_1,\Delta_2),a=i)=(1-P_2^N(a))P_1^N(a),\\	
	&P((1,1)|(\Delta_1,\Delta_2),a=i)=(1-P_1^N(a))(1-P_2^N(a)),\\
	&P((\Delta_1+1,\Delta_2+1)|(\Delta_1,\Delta_2),a=i)=P_1^N(a)P_2^N(a),
	\end{aligned}
	\end{equation}where $P_1^N(a)$ and $P_2^N(a)$ are the outage probability of client $1$ and client $2$, respectively, using NOMA with $\alpha_1=1-\frac{a}{N}$ and $\alpha_2=\frac{a}{N}$. The time superscript from the state $(\Delta_{1t},\Delta_{2t})$ and action $a_t$ is omitted for brevity.
	\item  $r: \mathcal{S} \times \mathcal{A}  \rightarrow \mathrm{R}$ is the one-stage reward function of state-action pairs, defined by $r(s_t,a)=w_1\Delta_{1t}+w_2\Delta_{2t}$.
 \end{itemize}
Given any initial state $s$, the infinite-horizon average reward of any feasible policy $\pi \in \Pi$ can be expressed as 
\begin{equation}
\label{e4}
C(\pi,s)=\lim_{T \rightarrow \infty}\sup\frac{1}{T} \sum_{k=0}^{T}{\mathbb E}_{s}^\pi[r(s_k,a_k)].
\end{equation} The Problem \ref{p1} can be transformed to the following MDP problem
\begin{prob}
	\label{p2}
	\begin{equation}
	\label{pro2}
	\min_{\pi} C(\pi,s).
	\end{equation}
\end{prob} 
To proceed, we now investigate the existence of optimal stationary and deterministic policy of Problem \ref{p2} and achieve the following theorem.
\begin{theorem}
	\label{TE}
	There exist a constant $J^{*}$, a bounded function $h(\Delta_1,\Delta_2):\mathcal{S} \rightarrow \mathcal{R}$ and a stationary and deterministic policy $\pi^{*}$, satisfies the average reward optimality equation,
	\begin{equation}
	\label{e10}
	J^{*}+h(\Delta_1,\Delta_2)=\min_{a\in \mathcal{A}} (w_1\Delta_1+w_2\Delta_2 \mathbb  +{\mathbb{E}}[h(\hat{\Delta}_1,\hat{\Delta}_2)]),
	\end{equation} $\forall (\Delta_1,\Delta_2) \in \mathcal{S}$, where $\pi^{*}$ is the optimal policy, $J^{*}$ is the optimal average reward, and $(\hat{\Delta}_1,\hat{\Delta}_2)$ is the next state after $(\Delta_1,\Delta_2)$ taking action $a$.  
\end{theorem}
\begin{proof}
	See Appendix \ref{A1}.
\end{proof}	
\subsection{Action Elimination}
In this subsection, we establish action elimination by analyzing the property of the formulated MDP problem, which can reduce action space for each state for lower computation complexity. According to \eqref{outage2} and \eqref{outage1}, and the fact $\alpha_1+\alpha_2=1$, the outage probability of client $2$ using NOMA $P_2^N$ is decreasing in $\alpha_2$, i.e., $P_2^N(a)$ is decreasing in action $a$, when $ \max\{\lceil{\frac{N}{2}\rceil}+1,\lceil\frac{(2^R-1)N}{2^R}\rceil\}<a<N$. However, the outage probability of client $1$ using NOMA $P_1^N$ is decreasing when $\frac{2^R-1}{2^R}<\alpha_2<\frac{2^R}{2^R+1}$ and increasing when $\frac{2^R}{2^R+1}<\alpha_2<1$. That is, $P_1^N(a)$ is decreasing in $a$ when $a\in \{\max\{\lceil{\frac{N}{2}\rceil}+1,\lceil\frac{(2^R-1)N}{2^R}\rceil\},...,\lfloor\frac{2^RN}{2^R+1}\rfloor\}$ and increasing when $a\in\{\lceil\frac{2^RN}{2^R+1}\rceil,\lceil\frac{2^RN}{2^R+1}\rceil+1,..., N-1\}$. As the BS aims to minimize the weighted sum of the expected AoI of the network, action  $a=\lfloor\frac{2^RN}{2^R+1}\rfloor$  has a better performance in reducing AoI of both clients, with lower outage probability comparing to $a\in \{\max\{\lceil{\frac{N}{2}\rceil}+1,\lceil\frac{(2^R-1)N}{2^R}\rceil\},\max\{\lceil{\frac{N}{2}\rceil}+1,\lceil\frac{(2^R-1)N}{2^R}\rceil\}+1,...,\lfloor\frac{2^RN}{2^R+1}\rfloor\}$. Thus, the action space can be reduced to $a\in\{0,\lfloor\frac{2^RN}{2^R+1}\rfloor,\lfloor\frac{2^RN}{2^R+1}\rfloor+1,..., N\}$.
\subsection{Structural Results on Optimal Policy}
In this subsection, we derive two structural results of the optimal policy which offers an effective way to reduce the offline computation complexity and online implementation hardware requirement.
\begin{theorem}
	\label{T2}
	The optimal policy $\pi^*$ has a switching-type policy. That is, denoting $c$ and $d$ as any action from action space $\{0,\lfloor\frac{2^RN}{2^R+1}\rfloor,\lfloor\frac{2^RN}{2^R+1}\rfloor+1,..., N\}$, 
	\begin{itemize}
		\item If $\pi^*((\Delta_1,\Delta_2))=c$, then $\pi^*((\Delta_1,\Delta_2+z))=d$, where $z$ is any positive integer and $d\geq c$,
		\item If $\pi^*((\Delta_1,\Delta_2))=a$, then $\pi^*((\Delta_1+z,\Delta_2))=d$, where $z$ is any positive integer and $d\leq c$.
	\end{itemize}
\end{theorem}
\begin{proof}
	See Appendix \ref{A2}.
\end{proof}	

Given the structure of the optimal policy, only the decision switching boundary is needed for implementation, rather than storing each state-action pair in the optimal policy, which significantly reduces the memory for the hardware. In addition, based on the structure, a special algorithm can be developed accordingly as in \cite{wang2018skip} to reduce the complexity of calculating the optimal policy.

\subsection{Suboptimal Policy}
In this subsection, we propose a suboptimal policy, with lower computation complexity comparing with that of optimal MDP policy. Inspired by the Max-weight policy in \cite{kadota2018scheduling}, the proposed suboptimal policy makes use of the transition probability in the MDP and minimizes the expectation of the reward of the next stage. According to \eqref{e3}, given current state $s=(\Delta_1,\Delta_2)$, the expected reward of next stage $\hat{s}$ can be expressed as 
\begin{small}
\begin{equation}
\mathbb{E}[r(\hat{s},a)]=\left\{
\begin{array}{rcl}
1+w_1P_1^O\Delta_1+w_2\Delta_2,& &\text{if }  a=1;\\
1+w_1\Delta_1+w_2P_2^O\Delta_2,& &\text{if } a=N\\
1+w_1P_1^N(a)\Delta_1+w_2P_2^N(a)\Delta_2,& &\text{otherwise.}  \\
\end{array}
\right.
\end{equation}
\end{small}Then, the action of state $s$ in suboptimal policy $\bar{\pi}$ can be given by
\begin{equation}
\bar{\pi}(s)=\arg \min_a \mathbb{E}[r(\hat{s},a)].
\end{equation} The suboptimal policy is simple and easy to implement. Moreover,  as we show via the numerical results in Section IV, the suboptimal policy can achieve near-optimal performance.  

By applying the suboptimal policy, the considered MDP problem is transferred into a Markov chain. We then can approximately calculate the weighted sum of the expected AoI of the network by finite state approximation as following,
\begin{equation}
P_{s,\hat{s}}^{\bar{\pi}}=\left\{
\begin{array}{rcl}
P_{s,\hat{s}}^{\bar{\pi}}, & &\text{if }  \hat{s}\in \mathcal{S}_m; \\
P_{s,\hat{s}}^{\bar{\pi}}+\sum_{s'\in \mathcal{S}-\mathcal{S}_m}P_{s,s'}^{\bar{\pi}}& & \text{otherwise, } \\
\end{array}
\right.
\end{equation} where $\mathcal{S}_m$ is the set of states with $\Delta_1\leq m$ and $\Delta_2\leq m$ and $P_{s,\hat{s}}^{\bar{\pi}}$ is the probability of transiting from state $s$ to $\hat{s}$ by taking action $\bar{\pi}(s)$. Accordingly, the transition matrix of the corresponding Markov chain $\mathbf{P}^{\bar{\pi}}$ can be constructed. Denoting $\pmb{\theta}$ as the steady state distribution of the Markov chain, by solving $\pmb{\theta}=\mathbf{P}^{\bar{\pi}}\pmb{\theta}$, the approximated weighted sum of the expected AoI of the network can be expressed by
\begin{equation}
\bar{\boldmath \Delta}'(\bar{\pi})=\sum_{s=(\Delta_1,\Delta_2)\in \mathcal{S}_m}\theta_s(w_1\Delta_1+w_2\Delta_2),
\end{equation} where $\theta_s$ is the element in $\pmb{\theta}$, denoting steady state probability of state $s$. This also serves as the lower bound of the weighted sum of the expected AoI of policy $\bar{\pi}$, i.e., $\bar{\boldmath \Delta}(\bar{\pi}) > \bar{\boldmath \Delta}'(\bar{\pi})$.

\section{Numerical Results}
This section provides numerical results to verify the analytical results provided in the preceding sections. We set path loss exponent $\tau=2$ and target data rate $R=1$ in all simulations. The SNR in this section refers to as the transmission SNR $\rho$.

We apply Relative Value Iteration (RVI) on truncated finite states ($\Delta_i \leq 100$, $\forall i$) to approximate the countable infinite state space according to \cite{sennott2009stochastic}. The optimal policy and suboptimal policy is illustrated in Fig.\ref{fig1}, where SNR$=18$dB, $d_1=2$, $d_2=4$ and $w_1=w_2=0.5$. The switching-type policy is verified. Besides, we can find that the proposed suboptimal policy is close to the optimal policy.
\begin{figure}[!t]
	\centering
	\begin{subfigure}[b]{0.22\textwidth}
		\centering
		\includegraphics[width=\textwidth]{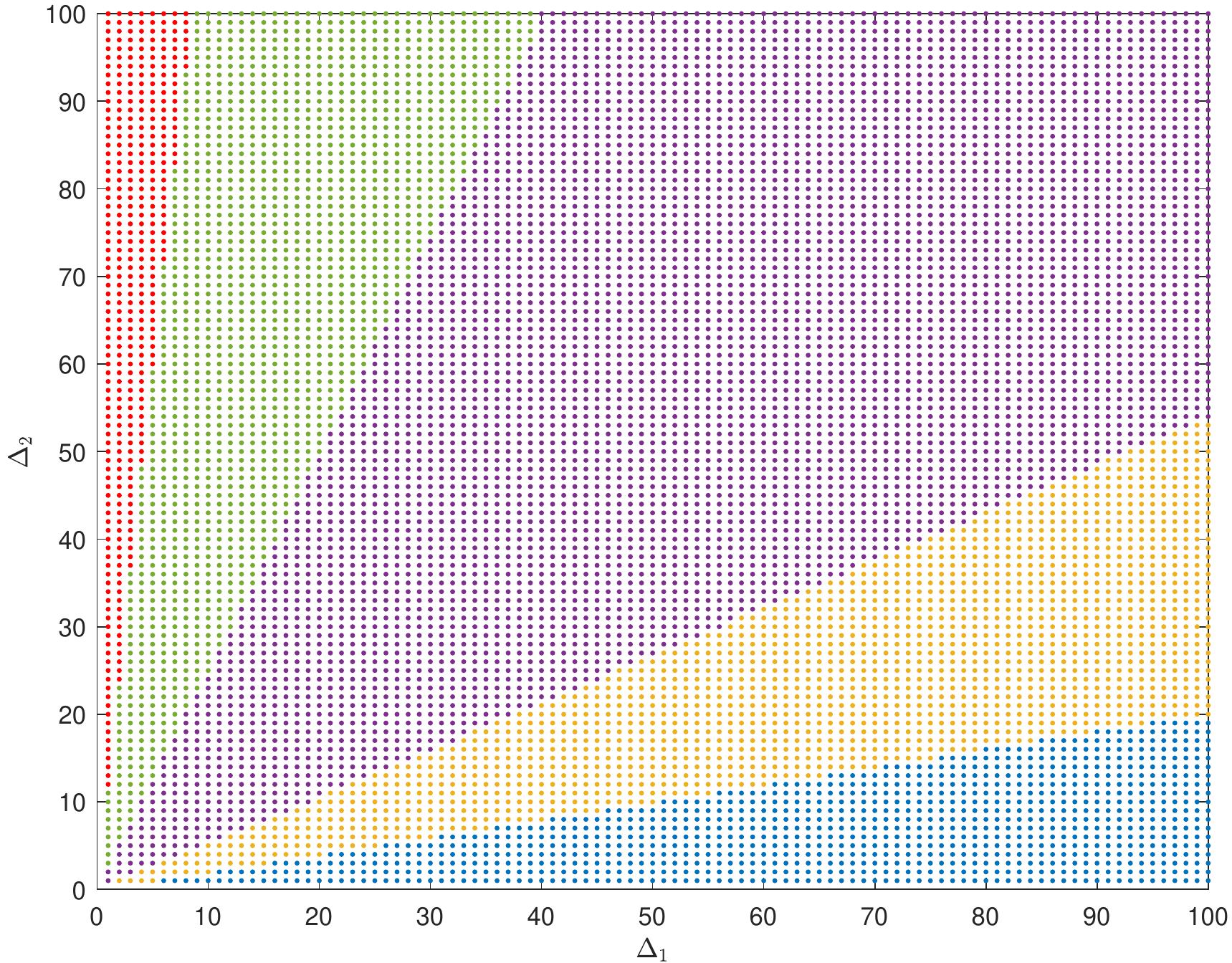}
		\caption{MDP optimal policy}
		\label{fig10}
	\end{subfigure}
	\hfill
	\begin{subfigure}[b]{0.22\textwidth}
		\centering
		\includegraphics[width=\textwidth]{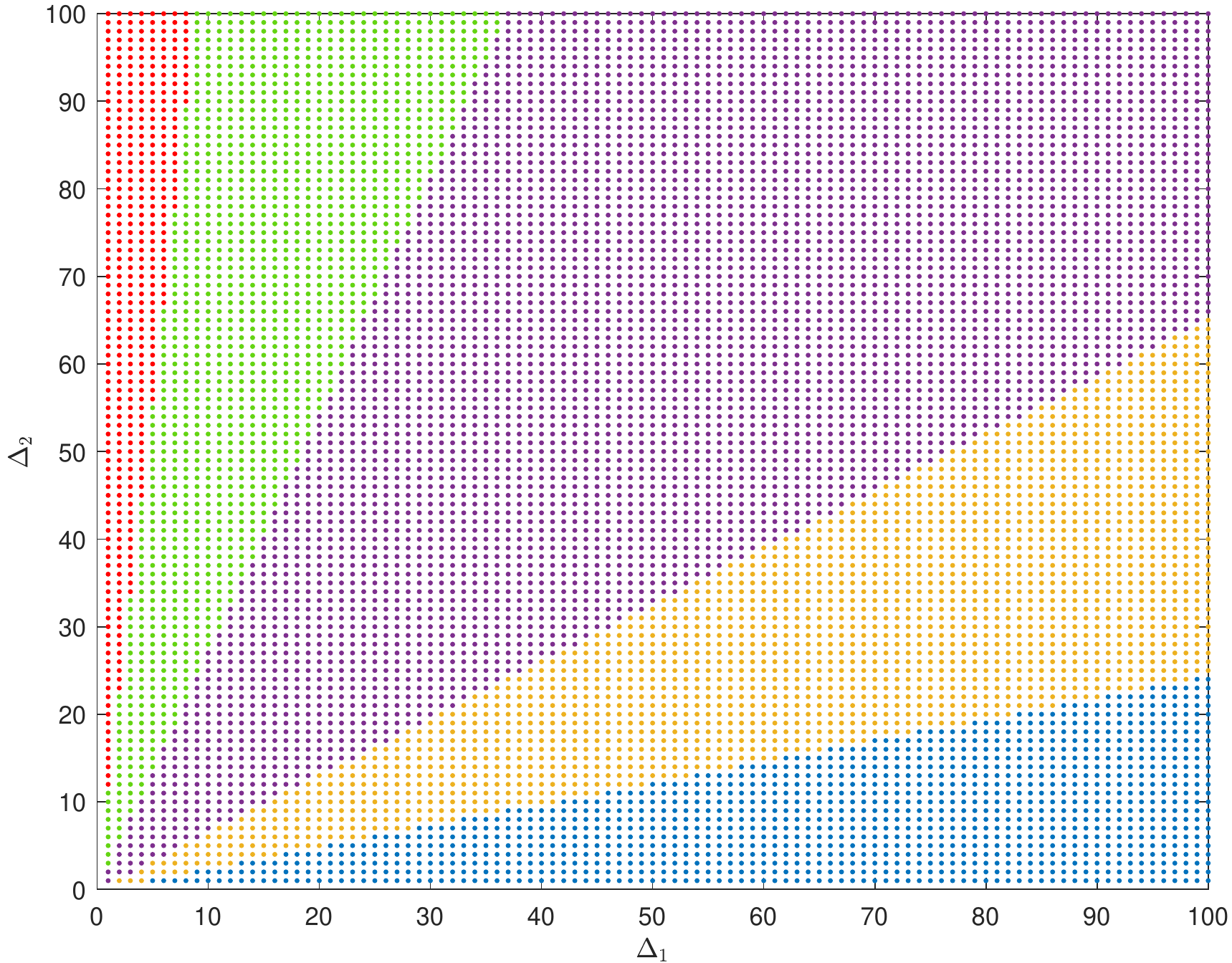}
		\caption{Suboptimal policy}
		\label{fig11}
	\end{subfigure}
	\caption{Age-optimal policy and suboptimal policy. Each point represents a state $s=(\Delta_1,\Delta_2)$. The colored area indicates action for each state, i.e., $a=0$ for states in the blue area; $a=7$ for states in the orange area; $a=8$ for states in the purple area; $a=9$ for states in the green area and $a=10$ for states in the red area, where $N=10$ and $\mathcal{A}=\{0,6,7,8,9,10\}$.}
	\label{fig1}
\end{figure}
\begin{figure}[!t]
	\centerline{\includegraphics[width=0.39\textwidth]{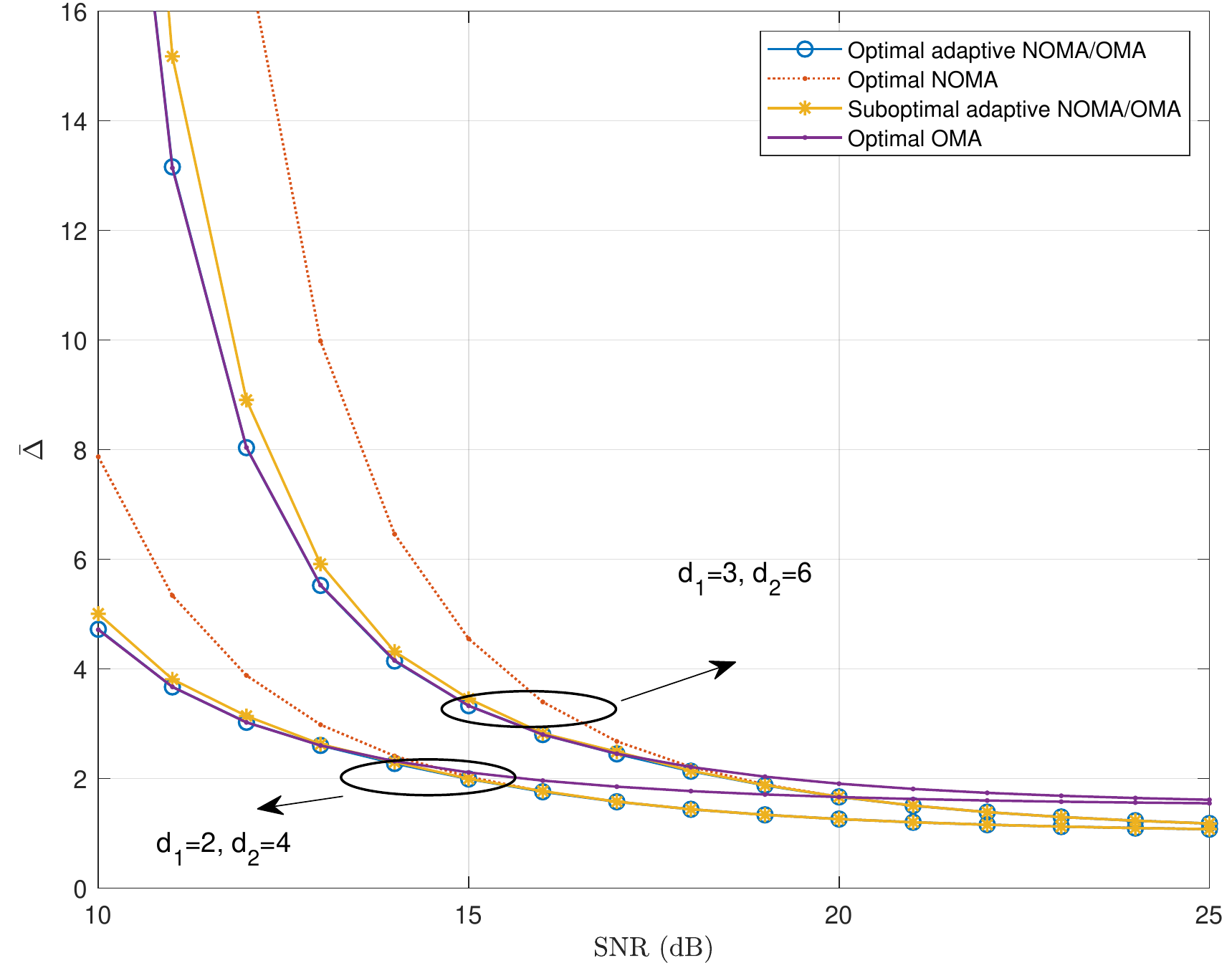}}
	\caption{The performance comparison of different policies versus SNR, when $w_1=w_2=0.5$.}
	\label{fig2}
\end{figure}

Fig. \ref{fig2} illustrates the performance of the weighted sum of the expected AoI of the two clients under optimal policy using adaptive NOMA/OMA scheme (optimal adaptive NOMA/OMA scheme), the policy that always using NOMA for transmission (optimal NOMA policy with $a\in\{\max\{\lceil{\frac{N}{2}\rceil}+1,\lceil\frac{(2^R-1)N}{2^R}\rceil\},...,N-1\}$), the proposed suboptimal policy and the optimal OMA policy that the BS adaptively selects one client to conduct transmission (optimal OMA scheme with $a\in \{0,N\}$) in two cases: 1) $d_1=2$ and $d_2=4$; 2) $d_1=3$ and $d_2=6$. The setting of the rest system parameters is the same as that in Fig \ref{fig1}. We can find that the proposed suboptimal policy achieves near-optimal performance, and its weighted sum of the expected AoI almost coincides with that of the optimal adaptive NOMA/OMA policy especially when the outage probability of two clients are small as shown in Fig. \ref{fig2}. Specifically, the performance of suboptimal policy is closer to that of the optimal adaptive NOMA/OMA policy when $d_1=2$ and $d_2=4$, comparing to the case when $d_1=3$ and $d_2=6$; the gap between the weighted sum of the expected AoI of the suboptimal policy and that of the optimal adaptive NOMA/OMA policy narrows as the SNR increases.

Moreover, we can see that when SNR is small, e.g., SNR$< 15$dB, the performance of optimal adaptive NOMA/OMA scheme and that of optimal OMA scheme are almost the same in Fig. \ref{fig2}. This is due to the low SNR, which leads to a larger outage probability for both OMA and NOMA. The situation for NOMA is even worse. As such, both optimal adaptive NOMA/OMA policy and the suboptimal policy will prefer not to choose NOMA scheme but use OMA scheme. Thus, these two policies have similar performance. As SNR increases, the weighted sum of the expected AoI of optimal OMA policy will approach $1.5$, when $w_1=w_2=0.5$. This is the optimal performance under the OMA scheme, where the outage probability of each client is 0, and thus $1$ and $2$ take turns to form the age evolution of each client. 

Furthermore, we can find that the performance of optimal adaptive NOMA/OMA policy and that of suboptimal policy and NOMA policy are relatively close when SNR is large, e.g., SNR$\geq20$dB in Fig. \ref{fig2}. This is because both optimal adaptive NOMA/OMA policy and suboptimal policy are more likely to choose NOMA for transmission. When SNR is large enough, the optimal performance of both the optimal adaptive NOMA/OMA policy and the suboptimal policy approaches $1$ as the instantaneous AoI of each client will be always $1$, thanks to no outage for both clients in NOMA. The BS thus always chooses NOMA scheme to conduct transmissions to both clients. In addition, NOMA is better than optimal OMA when SNR$>16$dB for $d_1=2$ and $d_2=4$ and SNR$>19$dB for  $d_1=3$ and $d_2=6$. This shows the advantage of NOMA in timely status update when SNR is large.


\section{Conclusions}
In this paper, we considered a wireless network with a base station (BS) conducting timely transmission to two clients in a slotted manner via hybrid non-orthogonal multiple access (NOMA)/orthogonal multiple access (OMA). The BS can adaptively switch between NOMA and OMA for the downlink transmission to minimize the AoI of the network. We develop an optimal policy for the BS to decide whether to use NOMA or OMA for downlink transmission based on the instantaneous AoI of both clients in order to minimize the weighted sum of the expected AoI of the network. This is achieved by formulating a Markov Decision Process (MDP) problem. We proved the existence of optimal stationary and deterministic policy. Action elimination was conducted to reduce the computation complexity. The optimal policy is shown to have a switching-type property with obvious decision boundaries. A suboptimal policy with lower computation complexity was also proposed, which is shown to achieve near-optimal performance according to simulation results. The approximate average AoI performance of the suboptimal policy was also derived. The performance of different policies under different system settings was compared and analyzed in numerical results to provide useful insights for practical system designs. 

\begin{appendices}
	\section{Proof of Theorem \ref{TE}}
	\label{A1}
	We prove this theorem by verifying the Assumptions 3.1, 3.2 and 3.3 in \cite{guo2006average} hold. As the action space for each state is finite, assumption 3.2 holds, and we only need to verify the following two conditions. 
	\begin{itemize}
		\item[\text{1)}] There exist positive constants $\beta < 1$, $M$ and $m$, and a measurable function $\omega(s) \geq1$ on $S$, $s=(\Delta_1,\Delta_2)$ such that the reward function of MDP problem $r(s,a)=w_1\Delta_1+w_2\Delta_2$, $|r(s,a)|\leq M\omega(s)$ for all state-action pairs $(s,a)$ and
		\begin{small}
			\begin{equation}
			\sum_{\hat{s}\in S} \omega(\hat{s})P(\hat{s}|s,a)\leq\beta\omega(s)+m,\ {\rm \ for \ all}\ (s,a).
			\end{equation}
		\end{small} 
		\item [\text{2)}] There exist two value functions $v_1,v_2 \in B_{\omega}(S)$, and some state $s_0\in S$, such that 
		\begin{small}
			\begin{equation}
			v_1(s)\leq h_{\alpha}(s)\leq v_2(s), \ {\rm for \ all} \ s\in S,\ {\rm and} \ \alpha \in(0,1),
			\vspace{-.5em}
			\end{equation} 
		\end{small}where {\small$h_{\alpha}(s)=V_{\alpha}(s)-V_{\alpha}(s_0)$} and {\small$B_{\omega}(S):=\{u:\Vert u\Vert_{\omega} <\infty \}$} denotes Banach space, {\small$\Vert u\Vert_{\omega}:=\sup_{s\in S}\omega(s)^{-1}|u(s)|$} denotes the weighted supremum norm.
	\end{itemize}
For condition 1, we show that when $\omega(s)=w_1\Delta_1+w_2\Delta_2$ and $m >1 $, there exist {\small$\max_a\{\frac{w_1\Delta_1 P_1^O+w_2\Delta_2+1-m}{w_1\Delta_1+w_2\Delta_2},\frac{w_1\Delta_1+w_2P_2^O\Delta_2+1-m}{w_1\Delta_1+w_2\Delta_2},$ $\;\frac{w_1P_1^N(a)\Delta_1+w_2P_2^N(a)\Delta_2+1-m}{w_1\Delta_1+w_2\Delta_2}\}\leq \beta<1$} to meet condition 1. To prove condition 2 in our problem, we show that when $\omega(s)=w_1\Delta_1+w_2\Delta_2$, there exists $\frac{w_1\Delta_1+w_2\Delta_2+1}{w_1\Delta_1+w_2\Delta_2}\leq \kappa<\infty$ that $\sum_{\hat{s}\in S} \omega(\hat{s})P(\hat{s}|s,a)\leq\kappa\omega(s) $ for all $(s,a)$, and for $d \in {\small \Pi^{MD}}$, ${\small \sum_{\hat{s}\in S} \omega(\hat{s})P_d(\hat{s}|s,a)\leq \omega(s)+1\leq (1+1)\omega(s)}$, so that  $\alpha^N\sum_{\hat{s}\in S} \omega(\hat{s})P_d^N(\hat{s}|s,a)\leq \alpha^N (\omega(s)+N)<\alpha^N(1+N)\omega(s)$. Hence, for each $\alpha$, $0\leq \alpha <1$, there exists a $\eta$, $0\leq \eta<1$ and an integer $N$ such that 
	\begin{equation}
	\alpha^N\sum_{\hat{s}\in S} \omega(\hat{s})P_{\pi}^N(\hat{s}|s,a)\leq \eta \omega(s)
	\end{equation} for $\pi=(d_1,...,d_N)$, where $d_k\in D^{MD}$, $1\leq k\leq N$. Then, according to Proposition 6.10.1 \cite{puterman2014markov}, for each $\pi\in \Pi^{MD}$ and $s\in S$
	\begin{small}
		\begin{equation}
		|V_{\alpha}(s)|\leq \frac{1}{1-\eta}[1+\alpha\kappa+...+(\alpha\kappa)^{(N-1)}]w(s).
		\vspace{-.5em}
		\end{equation}
	\end{small} We thus can further prove the condition 2. This completes the proof.

\section{Proof of Theorem \ref{T2}}
\label{A2}

The switching-type policy is actually the same as the monotone nondecreasing policy in $\Delta_2$ when $\Delta_1$ is fixed, and the monotone nonincreasing policy in $\Delta_1$ when $\Delta_2$ is fixed. To prove the monotonicity of the optimal policy of the MDP problem in $\Delta_2$, we verify that the following four conditions given in \cite[Theorem~8.11.3]{puterman2014markov} hold.
\begin{itemize}
	\item [a)] $r(s,a)$ is nondecreasing in $s$ for all $a\in A$;
	\item [b)] $q(k|s,a)=\sum_{j=k}^{\infty}p(j|s,a)$ is nondecreasing in $s$ for all $k\in S$ and $a\in A$;
	\item [c)] $r(s,a)$ is a subadditive function on $S\times A$ and
	\item [d)] $q(k|s,a)$ is a subadditive function on $S\times A$ for all $k\in S$.
\end{itemize}
To verify these conditions, we first order the state by $\Delta_2$, i.e., $s^+\geq s^-$ if $\Delta_2^+\geq \Delta_2^-$ where $s^+=(\cdot,\Delta_2^+)$ and $s^-=(\cdot,\Delta_2^-)$. The one-step reward function of the MDP is 
\begin{equation}
\label{req1}
r(s,a)=w_1\Delta_1+w_2\Delta_2.
\end{equation} It is obvious that the condition a) is satisfied. According to the transition probabilities in \eqref{e3} and \eqref{e3*}, if the current state $s=(\Delta_1,\Delta_2)$, the next possible states are $s_1=(\cdot,\Delta_2+1)$ (including $(1,\Delta_2+1)$ and $(\Delta_1+1,1)$) and $s_2=(\cdot,1)$ (including $(1,1)$ and $(\Delta_1+1,1)$).  Based on \eqref{e3} and \eqref{e3*}, we have
\begin{equation}
\label{req2}
q(k|s,a=0)=\left\{
\begin{array}{rcl}
0,& &\text{if }  k > s_1 \\
1,& &\text{otherwise.}  \\
\end{array}
\right.
\end{equation}
\begin{equation}
\label{req4}
q(k|s,a=i, 0<i<N)=\left\{
\begin{array}{rcl}
0,& &\text{if }  k > s_1\\
P_2^N(i),& &\text{if }  s_1\geq k>s_2 \\
1,& & \text{if } k\leq s_2
\end{array}
\right.
\end{equation}
\begin{equation}
\label{req3}
q(k|s,a=N)=\left\{
\begin{array}{rcl}
0,& &\text{if }  k > s_1\\
P_2^O,& &\text{if }  s_1\geq k>s_2 \\
1,& & \text{if } k\leq s_2
\end{array}
\right.
\end{equation} Thus, condition b) is immediate. 

To verify the remaining two conditions, we give the definition of subadditive in the following
\begin{definition}
	\label{d0}
	(Subadditive\cite{puterman2014markov}) A multivariable function $Q(\delta,a): \mathcal{N} \times \mathcal{A} \rightarrow R$ is subadditive in $(\delta,a)$ , if for all $\delta^{+}\geq\delta^{-}$ and $a^{+}\geq a^{-}$,
	\begin{equation}
	\label{e8}
	\begin{aligned}
	Q(\delta^{+},a^{+})+ Q(\delta^{-},a^{-}) \leq Q(\delta^{+},a^{-})+ Q(\delta^{-},a^{+})
	\end{aligned}
	\end{equation}holds.
\end{definition} 
According to \eqref{req1}, condition c) follows. For the last condition, we verify whether
\begin{equation}
q(k|s^+,a^+)+q(k|s^-,a^-)\leq q(k|s^+,a^-)+q(k|s^-,a^+),
\end{equation} with $s^+=(\Delta_1,\Delta_2^+)$ and $s^-=(\Delta_1,\Delta_2^-)$ where $\Delta_2^+\geq \Delta_2^-$ and $a^+\geq a^-$. As there are three actions, we consider three cases: (1) $a^+=i$, $a^-=0$, (2) $a^+=N$, $a^-=0$ and (3) $a^+=N$, $a^-=i$ and (4) $a^+=i$, $a^-=j$ for  $ 0\leq j\leq i \leq N$, $\forall i,j$.
According to \eqref{req2}-\eqref{req3}, we can verify that condition d) holds. As all these four conditions hold, the optimal policy is monotone nondecreasing in $\Delta_2$, when $\Delta_1$ is fixed. The proof of monotonicity of the optimal policy of the MDP problem in $\Delta_1$ is similar, thus omitted for brevity. This completes the proof.
\end{appendices}

\bibliography{ref}
\bibliographystyle{IEEEtran}

\end{document}